\documentclass[a4paper]{easychair}
\input{minted.cfg}

\usepackage[utf8]{inputenc}
\usepackage{xcolor}
\usepackage[\MINTEDOPTIONS]{minted}
\usepackage{amsmath}
\usepackage{amsthm}
\usepackage{amssymb}
\usepackage{stmaryrd} % for \llbracket and \rrbracket
\usepackage{bussproofs}
\usepackage{hyperref}
\usepackage{enumitem}
\usepackage{multicol}

\usepackage{mathabx} % \Dashv

\usepackage{mathpartir}

\usepackage{float}
\floatstyle{boxed}
\restylefloat{figure}

\definecolor{bg}{rgb}{0.95,0.95,0.95}
\setminted[ocaml]{bgcolor=bg}
\newcommand{\caml}[1]{\smash{\mintinline{ocaml}{#1}}}

\newcommand{\ind}[0]{\ensuremath{\mathtt{Ind}}}
\newcommand{\sep}[0]{\ensuremath{\mathtt{Sep}}}
\newcommand{\deepsep}[0]{\ensuremath{\mathtt{Deepsep}}}
\renewcommand{\vec}[1]{\ensuremath{\overline{#1}}}

\newcommand{\of}[0]{\ensuremath{\text{ of }}}

\newcommand{\guard}[2]{#2 \restriction #1}
\newcommand{\tyeq}[2]{(#1 = #2)}

\newcommand{\bnfeq}{\mathrel{::=}\;}
\newcommand{\bnfor}{\mathrel{\vert}}

\newcommand{\mcompsymbol}{\circ}
\newcommand{\mcomp}[2]{#1 \mathrel{\mcompsymbol} #2}

\newcommand{\subsymbol}{\triangleleft}
\newcommand{\sub}[2]{#1 \mathrel{\subsymbol} #2}

\newcommand{\vdashdecl}{\vdash_{\mathsf{decl}}}

\newcommand{\builtin}[1]{\mathtt{#1}}

\newtheorem{definition}{Definition}
\newtheorem{theorem}{Theorem}
\newtheorem{conjectured_theorem}[theorem]{Conjectured Theorem}
\newtheorem{lemma}{Lemma}
\newtheorem{conjectured_lemma}[lemma]{Conjectured Lemma}
\newtheorem{fact}{Fact}
\newtheorem{corollary}{Corollary}

\usepackage{natbib}
\setlist[itemize]{noitemsep, topsep=2pt}

\title{Unboxing Mutually Recursive Type Definitions in OCaml}
\titlerunning{Unboxing Mutually Recursive Type Definitions}

\author{
  Simon Colin % FIXME institution and email?
  \and
  Rodolphe Lepigre\inst{1}
  \and
  Gabriel Scherer\inst{2}
}
\authorrunning{Colin, Lepigre and Scherer}

\institute{
  Inria, LSV, CNRS, ENS Paris-Saclay, France --
  \email{rodolphe.lepigre@inria.fr}
  \and
  Inria, France --
  \email{gabriel.scherer@inria.fr}
}

\begin{document}

\maketitle

\begin{center} October 2018 \end{center}

\begin{abstract}
  In modern OCaml, single-argument datatype declarations
  (variants with a single constructor, records with a single
  immutable field) can sometimes be ``unboxed''. This means that their
  memory representation is the same as their single argument, omitting
  an indirection through the variant or record constructor, thus
  achieving better memory efficiency.
  However, in the case of generalized/guarded algebraic datatypes
  (GADTs), unboxing is not always possible due to a subtle assumption about
  the runtime representation of OCaml values. The current correctness
  check is incomplete, rejecting many valid definitions, in particular
  those involving mutually-recursive datatype declarations.
  In this paper, we explain the notion of \emph{separability} as
  a semantic for the unboxing criterion, and propose a set of
  inference rules to check separability. From these inference rules, we
  derive a new implementation of the unboxing check that properly
  supports mutually-recursive definitions.
\end{abstract}

% global setting for {mathpar} blocks
\abovedisplayskip .5em

\section{Introduction}%%%%%%%%%%%%%%%%%%%%%%%%%%%%%%%%%%%%%%%%%%%%%%%%%%%%%%%%

% FIXME mention that the single constructor should have a single argument?
% FIXME mention that the single record field should be immutable?
% This could be put in a footnote, but this is probably not the place for
% such (rather insignificant) details.

Version 4.04 of the OCaml programming language, released in
November 2016, introduced the possibility to \emph{unbox}
single-constructor variants and single-immutable-field records. In other words,
a value inhabiting such a datatype is exactly represented at runtime as
the value that it contains, rather than as a pointer to a \emph{block}
containing a tag for the constructor and the contained value. The
removal of this indirection is called \emph{constructor unboxing}, or
\emph{unboxing} for short, and it allows for a slight improvement in
speed and memory usage.

In the current version of OCaml, unboxing must be explicitly requested
with \caml{[@@unboxed]} as shown in the following example:
\begin{minted}{ocaml}
type uid = UId of int [@@unboxed]
\end{minted}

One of the main interests of unboxing is that it allows the
incorporation of semantic type distinctions without losing runtime
efficiency -- the mythical zero-cost abstraction. For example, the
elements of \caml{uid} are distinct from those of \caml{int} for the
type checker, but they have the same runtime representation. Unboxing
resolves a tension between software engineering and performance.

Unboxing becomes even more interesting when it is combined with
advanced features such as existential types, with GADTs, or
higher-rank polymorphism, with polymorphic record fields, which are
otherwise only accessible in boxed form.
\begin{minted}{ocaml}
type 'a data = { name : string ; data : 'a }
type any_data = Any_data : 'a data -> any_data [@@unboxed]

type proj = { proj : 'a. 'a -> 'a -> 'a } [@@unboxed]
\end{minted}

\subsection{Unboxing and dynamic floating-point checks}

OCaml uses a uniform memory representation, one machine word for all
values. Multi-word data such as floating-point numbers, records or
arrays are represented by a word-sized pointer to a block on the
heap.

For local computations involving floating-point numbers, the compiler
tries to optimize by storing short-lived floating-point values
directly, without the pointer indirection; this is called
\emph{floating-point unboxing}, a different form of unboxing
optimization. The boxing indirection needs to be kept for values
passed across function boundaries, or passed through data structures
that expect generic OCaml values.

To avoid indirection costs on typical numeric computations, a specific
representation exists for floating-point arrays, in which
floating-point numbers are stored unboxed, as two consecutive words in
memory. The compiler uses this optimized representation when an array
is statically known to have type \caml{float array}. It also performs
a dynamic optimization when creating a new non-empty array: it
checks whether its first element is a (boxed) floating-point value,
and in that case uses the special floating-point array
representation. As a consequence, all writes to such an array need to
be unboxed first, and reads produce data in already-unboxed form,
which meshes well with the local floating-point unboxing
optimizations.

However, this optimization crucially relies on an underlying
assumption: the inhabitants, the values of any given type are either
all boxed floating-point values, or none of them are. We don't know of
a standard name for this property, so we call it \emph{separability}.

Indeed, if there were non-separable types, containing both
floating-point and non-floating-point values, then this optimization
would be unsound.  This is demonstrated by the following example,
which relies on the unsafe/forbidden casting function
\caml{Obj.magic}.
\begin{minted}{ocaml}
let despicable : float array = [| 0.0 ; Obj.magic 42 |]
(* Produces: "segmentation fault (core dumped)" *)
\end{minted}
The array of floating point numbers that is constructed here is stored
using the optimized representation. As a consequence, its elements
must be unboxed prior to being inserted into the actual array. Here,
the segmentation fault is precisely triggered when attempting to unbox
the value \caml{42}, which is not stored in a block but as an
immediate memory word: dereferencing it as a float pointer accesses
forbidden memory.

Although the above example requires unsafe features, a similar situation
also arises when using a single-constructor GADT whose parameter is
existentially quantified, and can hence contain either floating-point or
non-floating-point values.
\begin{minted}{ocaml}
type any = Any : 'a -> any
\end{minted}
The above datatype is the GADT formulation of the existential type
$\exists \alpha.\,\alpha$: it can contain any OCaml value. However,
those values are ``boxed'' under the \caml{Any} constructor. In
particular, the value \caml{Any 0.0} is \emph{not} directly
represented as a floating-point value, but as pointer to a block
stored on the heap that is tagged with the \caml{Any} constructor,
following by (the OCaml representation of) the floating-point number.

If the \caml{any} datatype were allowed to be unboxed, we would have
constructed a type breaking the separability assumption. The above example
of segmentation fault could then be reproduced by a well-typed program
as follows.
\begin{minted}{ocaml}
type any = Any : 'a -> any [@@unboxed]
(* The above type is rightfully rejected by OCaml, it cannot be unboxed. *)

let array = [| Any 0.0 ; Any 42 |]
\end{minted}

The current implementation correctly rejects this \caml{any} datatype, but it
also rejects valid unboxed definitions that would not introduce
non-separable types, when mutually-recursive datatypes are involved. The
following real-world example of an interesting definition that is wrongly
rejected was given by Markus Mottl.\footnote{\url{https://github.com/ocaml/ocaml/pull/606\#issuecomment-248656482}}
% FIXME "value" is colored as a keyword. Change name?
\begin{minted}{ocaml}
type (_, _) tree =
  | Root  : { mutable value : 'a; mutable rank : int } -> ('a, [`root ]) tree
  | Inner : { mutable parent : 'a node }               -> ('a, [`inner]) tree

and _ node = Node : ('a, _) tree -> 'a node [@@ocaml.unboxed]
(* The above type is incorrectly rejected by OCaml, it could be unboxed. *)
\end{minted}
Yet another example of such wrongful rejection was encountered
by the second author, in the context of the Bindlib library~\citep*{bindlib}.
It resembles the definitions of \caml{any_type} and \caml{'a data} that were
given earlier, but it is rejected as the two datatypes are defined
mutually-recursively \citep[for more details, see][Section 3.1]{bindlib}.

\subsection{The existing check}
\label{subsec:existing-check}

The existing implementation of the compiler check for
unboxing was implemented by Damien Doligez in 2016, when constructor
unboxing was introduced. It proceeds by inspecting the parameter of every
unboxed GADT constructor. If it is an existential variable \caml{'a}, the
definition is rejected. If it is a type expression whose value
representation is known, such as a function \caml{foo -> bar} or a product
\caml{foo * bar}, then it is separable and the definition is accepted.
The difficult case arises when the parameter is of the form
\caml{(foo, bar, ...) t}. For example, if the parameter has type \caml{'a t},
where \caml{'a} is an existential variable, then this definition must be
rejected if \caml{t} is defined as \caml{type 'a t = 'a}, but it can be
accepted if it is defined as \caml{type 'a t = int * 'a}, for example.

In the current check, \caml{(foo, bar, ...) t} is expanded according to the
definition of \caml{t}, and its expansion is checked recursively. In the case
where \caml{t} is an abstract datatype, the check correctly fails as soon as one
of the parameter is an existential. If \caml{t} is part of the same block of
mutually-recursive definitions, its definition may not be known yet, and the
check fails although it could have succeeded had the definition been known.
Finally, it is worth noting that because of recursive datatypes, the expansion
process may not terminate. As a consequence, there is a hard limit on the
number of expansions.\footnote{This limit was originally set too high, and it
had to be reduced to avoid type-checking slowdowns.}

\subsection{Our approach: inference rules for separability}

We propose to replace the current check using a ``type system'' for
separability. In other words, we introduce inference rules to approximate
the semantic notion of separability: a type is separable if the values it
contains are all floating-point numbers, or if none of them are.
This approach is very similar to how the variance of datatype parameters is
handled in languages with subtyping, including OCaml, both in theory and
in practice. In the case of variance, the ``types'' are annotations such as
\texttt{covariant} or \texttt{contravariant}, and the ``terms'' that
are being checked are types and datatype definitions.

Here, the ``types'' are \emph{separability modes} indicating whether the
corresponding type parameters need to be separable for the whole datatype to be
separable. Modes include $\sep$ (separable), the mode of types or parameters
that must be separable, and $\ind$ (indifferent), the mode of types or
parameters on which no separability constraint is imposed.\footnote{There is
actually a third mode $\deepsep$ (deeply separable) that will be explained in
the next section.} For example, in the case of
\caml{type 'a t = 'a * int}, the parameter \caml{'a} has mode $\ind$ since
the values it contains are all pairs, independently of the type used to
instantiate \caml{'a}. For \caml{type ('a, 'b) second = 'b}, the parameter
\caml{'a} has mode $\ind$, but the parameter \caml{'b} has mode $\sep$.
Indeed, if \caml{'b} is not instantiated with a separable type, then the
whole definition cannot be separable either. The separability behavior of
a parametrized datatype is characterized by a \emph{mode signature}, which gives
a choice of mode for the type parameters that guarantees that the whole datatype
will be separable. For instance, the previous two examples would have mode
signatures \caml{('a:Ind) t} and \caml{('a:Ind, 'b:Sep) second}.

If a type that is being checked is defined in terms of a type constructor
\caml{(foo, bar, ...) t}, and if the mode signature of \caml{t} is known,
then its definition does not need to be unfolded as in the legacy
implementation. Indeed, it is enough to simply check the parameter instances
(\caml{foo}, \caml{bar}, ...) against the modes of the signature. For example,
if we encounter \caml{(foo, bar) second}, then we  only need to check that
\caml{bar} is separable.

In the case of a block of mutually-recursive datatype definitions,
a mode signature must be constructed for each definition of the
block. However, this cannot be done separately for each definition due
to dependencies.  As a consequence, we proceed by computing
a fixpoint, as for variance: we iteratively refine an approximation of
the mode signatures for the block, updating when encountering
conflicts, until a mode signature that requires no update is found: it
is a valid signature for the block.
The number of possible modes assignments for each parameter is finite,
so this fixpoint computation always terminates.
Note that our inference rules do not talk about the fixpoint computation,
or about algorithmic aspects in general. It is a simpler, declarative
specification for the correctness of mode assignments, from which the
checking algorithm can be derived. It can also be used to reason about
the semantic correctness of our check.

\subsection{Contributions}

We claim the following three contributions:
\begin{itemize}
\item A clear explanation of separability. The notion of separability
  evidently existed in Damien Doligez's mind when unboxed datatypes
  where implemented in 2016, but we had to rediscover it to understand
  the check, and we took this opportunity to document separability and
  to give it a precise semantic definition.
\item A set of inference rules for separability of type expressions
  and datatype definitions.
\item An implementation of a separability check derived from these
  inference rules, within the OCaml compiler. It is compatible with
  mutually-recursive definitions, which were previously always rejected.
\end{itemize}
A preliminary version of this work was produced during an internship
of the first author~\citep*{report_colin}. We give here a full
treatment of GADTs, while the internship report used first-class
existential types, without equations. We also moved from a prototype
implementation, separate from the OCaml type-checker and defined on
simpler data-structures, to a production-ready implementation of the
separability check in (an experimental fork of) the OCaml type-checker.

\section{Type language and separability modes}%%%%%%%%%%%%%%%%%%%%%%%%%%%%%%%%

As mentioned in the introduction, our approach to unboxability checking
relies on the notion of \emph{separability} of a type, which was already
introduced in an intuitive way.
\begin{definition}[Separability]
  A type is said to be separable if and only if its inhabitant contain
  either only floating-point values, or only non-floating point values.
\end{definition}
Intuitively, our final goal is to make sure that all type definitions have
a separable body. That is, for all definition
\caml{type ('a, 'b, ...) t = <type_body>},
we want to check that the type \caml{<type_body>} is separable under
some assumptions on the separability of parameters. Before going into
formal definitions, two potential sources of difficulties must be
discussed: GADTs, which are the only possible source of
non-separability, and the (related) typing constraints, which can be
used to extract possibly non-separable subcomponents of a type.

\subsection{GADTs using type equalities and existential quantifiers}
\label{subsec:gadts-reminder}

Generalized/guarded algebraic datatypes, GADTs for short extend the
usual variant datatypes using a slightly different syntax. They are
parametrized datatypes \caml{(_, _, ...) t} in which the typing
constraints on the parameters may vary depending on the variant
constructor. Moreover, existentially quantified type variables may
appear in GADT constructors.\footnote{See the OCaml manual
  (\url{https://caml.inria.fr/pub/docs/manual-ocaml-4.07/extn.html\#sec252})
  for a more thorough introduction to GADTs.

  We are intentionally unclear about whether the ``G'' means
  ``generalized'' or ``guarded'', because both words have been used in
  the literature. In the present article, ``guarded'' makes more sense
  as we concentrate on the pains introduced by equality constraints,
  or equality \emph{guards}. ``Generalized'' is also a rather empty
  name, as there are many other ways to generalize algebraic
  datatypes.} Examples of GADTs illustrating these features are given
below.
\begin{minted}{ocaml}
type _ data =
  | Char : char -> char data
  | Bool : bool -> bool data

type any_function = Fun : ('a -> 'b) -> any_function

type _ first = First : 'c -> ('c * 'd) first
\end{minted}

As it turns out, GADTs can be decomposed into more primitive
components: non-GADT algebraic datatypes, \emph{type
  equalities} and \emph{existential types}. For instance, the above
examples can be encoded as follows, using an imaginary extension of
the OCaml syntax explained below.\footnote{The translation can be
  defined in a systematic way, see for example
  \citet*{simonet-pottier-hmg-toplas}.}
\begin{minted}{ocaml}
type 'a data =
  | Char of char with ('a = char)
  | Bool of bool with ('a = bool)

type any_function = Fun of exists 'a 'b. 'a -> 'b

type 'a first = First of exists 'b 'c. 'b with ('a = 'b * 'c)
\end{minted}
In the above, we use the new syntax \caml{exists 'a. <type_expr>} for
first-class existential quantification over one or several type
variables, and \caml{<type_expr> with ('a = <type_expr>)} for guarding
a type expression with an equality constraint. In our formal
description of the separability check, we will use the syntax
$\exists \alpha.\tau$ for existentials, and the syntax
$\guard{\tyeq \alpha \kappa}{\tau}$ for equality guards, where
$\alpha$ is a type variable, and $\tau$ and $\kappa$ are type
expressions.

The semantic intuition behind these two type-formers is the following. An
existential type $\exists \alpha.\tau$ can be understood as the union over
all types $\kappa$ of $\tau[\alpha := \kappa]$. For example, the values of
the type
\caml{exists 'a. 'a -> 'a}
contains the values of
\caml{int -> int},
but also the values of
\caml{bool -> bool}
and the values of
\caml{(char -> bool) -> (char -> bool)}.
\footnote{The interpretation of existential quantifiers as unions
  (and dually, of universal quantifiers as intersections) is very
  common in realizability semantics, for example.} An equality guard
$\guard{\tyeq \alpha \kappa}{\tau}$ exactly corresponds to $\tau$ in
the case where the constraint $\tyeq \alpha \kappa$ is satisfied, and
it is empty otherwise. Note that guards whose left-hand-side is a type
variable suffice to express GADTs: all guards equate a type parameter
to its instance in the ``return type'' of the variant constructor.

\subsection{Equality guards and deep separability}
\label{subsec:intro-deepsep}

In the introduction, the parametrized type constructors \caml{'a t}
that were considered either had the mode signature \caml{('a : Ind)
  t}, meaning that \caml{'a t} is always separable no matter what
\caml{'a} is, or they had mode signature \caml{('a : Sep) t}, meaning
that \caml{'a t} is separable only if \caml{'a} is itself
separable. However, these two modes are not always sufficient due to
equality guards.

The equality guards used in GADTs\footnote{Constrained datatype definitions
such as \caml{type 'a t = 'b constraint 'a = 'b * int} may also involve
equality guards, but we chose to ignore this feature here since it poses
exactly the same problem.} introduce the ability for parametrized types to
``peek'' into the definition of their parameters in ways that affect our
separability check.
Consider, for example, the unboxed version of the \caml{_ first} datatype, which
is accepted by the current constructor unboxing check.
\begin{minted}{ocaml}
type _ first =
  | First : 'b -> ('b * 'c) first [@@unboxed]
\end{minted}
Using this definition, \caml{('b * int) first} has the same memory
representation as \caml{'b}, so it is separable if and only if \caml{'b}
is separable. In other words, the separability of type \caml{foo first}
does not  on the separability of type \caml{foo}, but rather
% FIXME this is not really true, it does not depend at all on the separability
% of foo as a whole. But do we want to go into more details?
on the separability of a \emph{sub-component},
in the sense of syntactic inclusion, of the type \caml{foo}.

To account for this situation, we introduce a third separability mode
$\deepsep$ (deeply separable). For a closed type expression (a type
expression with no free type variables) to be $\deepsep$, all its
sub-components, including the type expression itself, must be separable.

\subsection{Formal syntax of types}

\begin{figure}[htbp]
  Syntax of type expressions:
  \begin{mathpar}
    \begin{array}{l@{~}r@{~}l@{\quad}l}
      \tau, \kappa
      & \bnfeq & \alpha, \beta & \mbox{type variable} \\
      & \bnfor & \builtin{float} \bnfor \builtin{int} \bnfor \builtin{bool}
               & \mbox{builtin types} \\
      & \bnfor & t(\tau_1, \dots, \tau_n) & \mbox{(parametrized) type constructor} \\
      & \bnfor & \tau \to \kappa & \mbox{function type} \\
      & \bnfor & \tau_1 \times \dots \times \tau_n & \mbox{product/record type} \\
      & \bnfor & \forall \alpha. \tau & \mbox{polymorphic type} \\
      & \bnfor & \exists \alpha. \tau & \mbox{existential type} \\
      & \bnfor & \guard{\tyeq \alpha \kappa}{\tau} & \mbox{equality guard} \\
  \end{array}
  \end{mathpar}
  \medskip

  Syntax of datatypes:
  \begin{mathpar}
    \begin{array}{l@{~}r@{~}l@{\quad}l}
      A, B
      & \bnfeq & C_1 \of \tau_1 \;|\; \dots \;|\; C_n \of \tau_n
      & \mbox{boxed variant} \\
      & \bnfor & C \of \tau \; \mathtt{[@@unboxed]}
      & \mbox{unboxed variant} \\
      & \bnfor & \{ \mathtt{mutable}?~l_1 : \tau_1
                    \;;\; \dots \;;\;
                    \mathtt{mutable}?~l_n : \tau_n \}
      & \mbox{boxed record} \\
      & \bnfor & \{l : \tau\} \; \mathtt{[@@unboxed]}
      & \mbox{unboxed record} \\
      & \bnfor & \tau
      & \mbox{type synonym} \\
    \end{array}
  \end{mathpar}
  \medskip

  Sub-component relation on type expressions:
  \belowdisplayskip 0em
  \begin{mathpar}[\mprset{lineskip=.8em}]
  \infer{ }{\sub \tau \tau}

  \infer
  {\sub {\tau_1} {\tau_2} \\ \sub {\tau_2} {\tau_3}}
  {\sub {\tau_1} {\tau_3}}
  \\
  \infer{i \in [1; n]}{\sub {\tau_i} {t(\tau_1, \dots, \tau_n)}}

  \infer{1 \in \{1, 2\}}{\sub {\tau_i} {\tau_1 \to \tau_2}}

  \infer{i \in [1; n]}{\sub {\tau_i} {\tau_1 \times \dots \times \tau_n}}
  \\
  \infer
    {\sub \tau \kappa \\ \alpha \notin \tau}
    {\sub \tau {\forall \alpha. \kappa}}

  \infer
    {\sub \tau \kappa \\ \alpha \notin \tau}
    {\sub \tau {\exists \alpha. \kappa}}

  \infer
    {i \in \{1, 2\}}
    {\sub {\tau_i} {(\guard{\tyeq \alpha {\tau_1}}{\tau_2})}}
  \end{mathpar}
  \caption{Syntax of type expressions, datatypes, and sub-component relation.}
  \label{fig:type-syntax}
\end{figure}

We define the syntax of the type expressions and datatypes that we consider
in the top and middle parts of Figure~\ref{fig:type-syntax}.
It is a representative subset of the OCaml grammar of
types,\footnote{For instance, we omit object types and polymorphic
  variants, but they could be handled just like products.} with
first-class existential types $\exists \alpha. \tau$ and type
restrictions $\guard{\tyeq \alpha \kappa}{\tau}$ to represent GADTs
with finer-grained rules, as explained in
Section~\ref{subsec:gadts-reminder}. We also include first-class
universal types $\forall \alpha. \tau$, although they are only allowed
in record or method fields in OCaml.

To define deep separability, we first need to define the syntactic
sub-components of a type (see Section~\ref{subsec:intro-deepsep}). We
define a sub-component relation $\sub \tau \kappa$ (``$\tau$ is
a sub-component of $\kappa$'') at the bottom of
Figure~\ref{fig:type-syntax}. The first two rules make the relation
reflexive and transitive, and the others give the immediate
sub-components for each type-former.

The definition of the sub-component relation is careful to preserve
the scoping of variables. For example,
$\sub \tau {\forall \alpha. \kappa}$ holds only if $\alpha$ does not
occur free in $\tau$, which we write $\alpha \notin \tau$. In
particular, whenever $\sub \tau \kappa$ holds, the free type variables
of $\tau$ are included in those of $\kappa$; the sub-components of
a closed type expression are also closed types.

\subsection{Separability modes}

\begin{figure}[htbp]
  \belowdisplayskip 0em
  \begin{multicols}{2}
    Separability modes (or simply, modes):
    \begin{mathpar}
      \begin{array}{l@{~}r@{~}l@{\quad}l}
        m, n & \bnfeq & \ind & \mbox{indifferent} \\
             & \bnfor & \sep & \mbox{separable} \\
             & \bnfor & \deepsep & \mbox{deeply separable}
      \end{array}
    \end{mathpar}

    Mode composition $\mcomp m n$:
    \begin{mathpar}
      \begin{array}{r@{\quad:=\quad}l}
        \mcomp \ind m & \ind \\
        \mcomp \sep m & m    \\
        \mcomp \deepsep m & \deepsep
      \end{array}
    \end{mathpar}
  \end{multicols}

  \begin{center}
    Order structure:\\
    $\ind < \sep < \deepsep$
  \end{center}
  \vspace{-.8em}
  \caption{Separability modes and mode operations.}
  \label{fig:separability-modes}
\end{figure}

Separability modes $m$, $n$, their order structure $m < n$ and mode
composition $\mcomp m {n}$ are defined in
Figure~\ref{fig:separability-modes}.
Modes are totally ordered from less to more demanding.
We may use derived notations such as $m \leq n$, $\min(m,n)$
or $\max(m,n)$ for the non-strict ordering, the minimum or the maximum of
two modes.

A mode $m$ expresses a requirement on a type expression, which comes from
its context: for the whole expression to be valid, some of its
sub-components must have the separability property $m$. The operation $\mcomp
m {n}$ expresses composition of those requirements in our inference rules
(given in the next section). For example, if $t(\alpha)$ requires its
parameter $\alpha$ to have mode $m$ for the whole expression to be separable,
and $u(\beta)$ requires its parameter $\beta$ to have mode $n$ for
the whole expression to be separable, then $t(u(\beta))$ is separable
when $\beta$ has mode $\mcomp m {n}$.

\subsection{Contexts and mode signatures}

Our separability judgments in Section~\ref{sec:separability_inference}
make use of contexts $\Gamma$, representing separability assumptions
on the type variables ($\alpha, \beta\dots$) that are in scope.
Moreover, we also rely on mode signatures $\Sigma$, representing the
separability requirements on the datatype constructors $t(\vec{\alpha})$
that are available in the scope. Contexts and mode signatures are defined in
Figure~\ref{fig:mode-signatures}.

\begin{figure}[htbp]
  \begin{mathpar}
      \Gamma \;\bnfeq\; \emptyset \mid \Gamma, \alpha : m

      \Sigma \;\bnfeq\; \emptyset \mid \Sigma, t(\vec{\alpha : m})
      \\
      \infer
      {\Gamma \leq \Gamma' \\ m \leq m'}
      {\Gamma, \alpha:m \leq \Gamma', \alpha:m'}

      \infer
      {\Sigma \leq \Sigma' \\ \vec{\alpha} = \vec{\alpha'} \\ \vec{m \geq m'}}
      {\Sigma, t(\vec{\alpha:m}) \leq \Sigma', t(\vec{\alpha':m'})}
  \end{mathpar}
  \caption{Contexts, mode signatures, and their order}
  \label{fig:mode-signatures}
\end{figure}

This figure also defines the extension of the order on modes to an
order on contexts and on mode signatures. The order on contexts is just
a point-wise extension of the mode order, imposing that two comparable
contexts have the same type variables. The order on signatures imposes
the same datatype constructors on both sides, but requires their
parameter modes to be pointwise comparable in the opposite order,
$(\geq)$ rather than $(\leq)$: parameters are in contravariant position.
For example, we have
$t(\alpha:\sep, \beta:\sep) \leq t(\alpha:\ind, \beta:\ind)$.

\section{Separability inference}
\label{sec:separability_inference}

The formal deduction rules that we will use to assess the separability of
datatype definitions are defined in Figure~\ref{fig:rules}. Type expressions
are checked by judgments of the form $\Sigma; \Gamma \vdash \tau : m$, which
can intuitively be read in either direction, from $\Gamma$ to $m$ or
conversely:
\begin{itemize}
\item If the type variables respect the modes in $\Gamma$,
  then the type expression $\tau$ has the mode $m$.
\item For the type expression $\tau$ to be safe at mode $m$, then
  its type variables need to have at least the modes given in $\Gamma$.
\end{itemize}

\begin{figure}[htbp]
  Inference rules at type-expression level:
  \begin{mathpar}[\mprset{lineskip=.8em}]
    % Type variable, type constructor.
    \AxiomC{$(\alpha : m) \in \Gamma$}
    \UnaryInfC{$\Sigma; \Gamma \vdash \alpha : m$}
    \DisplayProof

    \AxiomC{$t(\alpha_1 : m_1, \dots, \alpha_n : m_n) \in \Sigma$}
    \AxiomC{$(\Sigma; \Gamma \vdash \tau_i : \mcomp m {m_i})_{1 \leq i \leq n}$}
    \BinaryInfC{$\Sigma; \Gamma \vdash t(\tau_1, \dots, \tau_n) : m$}
    \DisplayProof

    % Function and product.
    \AxiomC{$\Sigma; \Gamma \vdash \tau : m \circ \ind$}
    \AxiomC{$\Sigma; \Gamma \vdash \kappa : m \circ \ind$}
    \BinaryInfC{$\Sigma; \Gamma \vdash \tau \to \kappa : m$}
    \DisplayProof

    \AxiomC{$(\Sigma; \Gamma \vdash \tau_i : m \circ \ind)_{1 \leq i \leq n}$}
    \UnaryInfC{$\Sigma; \Gamma \vdash \tau_1 \times \dots \times \tau_n : m$}
    \DisplayProof

    % Universal and existential quantifiers.
    \AxiomC{$\Sigma; \Gamma, \alpha : n \vdash \tau : m$}
    \UnaryInfC{$\Sigma; \Gamma \vdash \forall \alpha. \tau : m$}
    \DisplayProof

    \AxiomC{$\Sigma; \Gamma, \alpha : \ind \vdash \tau : m$}
    \UnaryInfC{$\Sigma; \Gamma \vdash \exists \alpha. \tau : m$}
    \DisplayProof

    % coercion
    \AxiomC{$\Sigma; \Gamma \vdash \tau : m$}
    \AxiomC{$m \geq n$}
    \BinaryInfC{$\Sigma; \Gamma \vdash \tau : n$}
    \DisplayProof

    % Guard
    \AxiomC{$
      \forall \Gamma' \geq \Gamma,
      \quad
      \Sigma; \Gamma' \vdash \tyeq {\kappa_1} {\kappa_2}
      \implies
      \Sigma; \Gamma' \vdash \tau : m
      $}
    \UnaryInfC{$\Sigma; \Gamma \vdash \guard{\tyeq {\kappa_1} {\kappa_2}}{\tau} : m$}
    \DisplayProof

    \infer
    {\forall m,\quad
     \Sigma; \Gamma \vdash \tau_1 : m
     \iff
     \Sigma; \Gamma \vdash \tau_2 : m}
    {\Sigma; \Gamma \vdash \tyeq {\tau_1} {\tau_2}}
  \end{mathpar}

  \medskip

  Inference rules at datatype level:
  % Type declaration (variants).
  \begin{mathpar}
    \AxiomC{}
    \UnaryInfC{$\Sigma; \Gamma \vdashdecl
      C_1 \of \tau_1 \;|\; \dots \;|\; C_n \of \tau_n$}
    \DisplayProof

    \AxiomC{$\Sigma; \Gamma \vdash \tau : \sep$}
    \UnaryInfC{$\Sigma; \Gamma \vdashdecl C \of \tau \; \mathtt{[@@unboxed]}$}
    \DisplayProof

    \AxiomC{}
    \UnaryInfC{$\Sigma; \Gamma \vdashdecl
      \{l_1 : \tau_1 \;;\; \dots \;;\; l_n : \tau_n\}$}
    \DisplayProof

    \AxiomC{$\Sigma; \Gamma \vdash \tau : \sep$}
    \UnaryInfC{$\Sigma; \Gamma \vdashdecl \{l : \tau\} \; \mathtt{[@@unboxed]}$}
    \DisplayProof

    \AxiomC{$\Sigma; \Gamma \vdash \tau : \sep$}
    \UnaryInfC{$\Sigma; \Gamma \vdashdecl \tau$}
    \DisplayProof
  \end{mathpar}

  \medskip

  Inference rule at the datatype declaration block level:
  % Type block declaration.
  \belowdisplayskip 0em
  \begin{mathpar}
    \AxiomC{$\exists (\vec{m_i})_{1 \leq i \leq n},\ \Sigma_{\mathsf{block}}
      := t_1(\vec{\alpha_1} : \vec{m_1}), \dots, t_n(\vec{\alpha_n} : \vec{m_n})$}
    \AxiomC{$\left(
        \Sigma_{\mathsf{env}}, \Sigma_{\mathsf{block}}
        ; \vec{\alpha_i} : \vec{m_i} \vdashdecl A_i
      \right)_{1 \leq i \leq n}$}
    \BinaryInfC{$\Sigma_{\mathsf{env}} \vdash
      (t_i(\vec{\alpha_i}) := A_i)_{1 \leq i \leq n}
      \dashv \Sigma_{\mathsf{block}}$}
    \DisplayProof
  \end{mathpar}

  \caption{Inference rules for separability.}
  \label{fig:rules}
\end{figure}

\paragraph{Type expressions.}

The rule for parametrized datatypes $t(\tau_1, \dots, \tau_m)$ uses
mode composition; for example, if $t$ is a one-argument type
constructor with signature $t(\alpha : n)$ and we want $t(\tau)$ to
have mode $m$, then $\tau$ needs to have the mode $\mcomp m {n}$.

The rules for concrete datatypes (functions and products, but also
most other OCaml type-formers if we were to add them) use
$\mcomp m \ind$ on their arguments. If $m$ is $\sep$ or $\ind$, then
$\mcomp m \ind$ is $\ind$, which corresponds to not requiring anything
of the sub-components: the values at type $\tau_1 \to \tau_2$ are all
functions, never floats, regardless of $\tau_i$. If $m$ is
$\deepsep$, then we do need to check the sub-components, and indeed
$\mcomp \deepsep \ind$ is $\deepsep$.

A value is at the universal type $\forall \alpha.\,\tau$ only if
it belongs to all the $\tau[\kappa/\alpha]$ for all $\kappa$. In
particular, it is enough to prove the separability at just one of
these $\tau[\kappa/\alpha]$, the universal type has even less values,
so we can assume the arbitrary mode $n$ of this particular $\kappa$
by adding $\alpha:n$ in the context. Conversely,
$\exists \alpha.\,\tau$ is inhabited by all the $\tau[\kappa/\alpha]$,
so $\tau$ has to have the desired mode for all possible modes of
$\kappa$. Instead of requiring the premise to hold for all possible
modes $n$, we equivalently ask for the most demanding mode $\ind$.

The conversion rule allows to forget some information about a type
$\tau : m$ by exporting it at a smaller mode $n \leq m$. For example, all
$\sep$ types are also $\ind$ types.

\paragraph{Equality constraints.}

The rule for equality constraints is the most complex rule in
the system. A first remark is that $\guard {\tyeq {\tau_1} {\tau_2}}
\kappa$ always has less elements than $\kappa$: it has the same elements
when the equality holds, or no elements otherwise.
In particular, if $\kappa : m$ holds in the current context $\Gamma$,
then $\guard {\tyeq {\tau_1} {\tau_2}} \kappa : m$
should also hold in $\Gamma$.

When we see an equality constraint, we gain more information, which
should allow us to mode-check more types. The way our rules represent
this information gain is by moving from the current mode context
$\Gamma$ to a stronger mode context $\Gamma'$.  More precisely, to
check $\guard {\tyeq {\tau_1} {\tau_2}} \kappa : m$ in $\Gamma$, the
rule asks to check $\kappa : m$ in \emph{any} context
$\Gamma' \geq \Gamma$ that is consistent with the assumption
$\tyeq {\tau_1} {\tau_2}$. This corresponds to the
$\Gamma' \vdash \tyeq {\tau_1} {\tau_2}$ hypothesis, which will be
explained shortly. For example, if $\Gamma$ is
$\alpha : \ind, \beta : \sep$, and we observe the equality
$\tyeq \alpha \beta$, then in particular we know that $\alpha : \sep$
also holds: if the two types are equal, they must have the same mode.
Our rule will type-check the body type $\kappa$ in
stronger contexts $\Gamma'$ with $\alpha : \sep, \beta : \sep$.

A context $\Gamma$ is valid for a type equality, written
$\Gamma \vdash \tyeq {\tau_1} {\tau_2}$, if the two types $\tau_1$ and
$\tau_2$ have exactly the same mode in $\Gamma$. Remark that it is not
enough to ask that, for a given mode $m$, both types have mode $m$;
for example, all types trivially have mode $\ind$. Instead, we ask
that for \emph{any} mode $m$, either $\tau_1$ and $\tau_2$ have mode
$m$, or neither of them have it. This is equivalent to requiring
$\tau_1$ and $\tau_2$ to have the same ``best'', maximal mode.

In our example where $\Gamma$ is $\alpha : \ind, \beta : \sep$, note that
we do not have $\Gamma \vdash \tyeq \alpha \beta$: $\beta$ has mode
$\sep$ but $\alpha$ does not. The modes $\Gamma' \geq \Gamma$ that
satisfy $\Gamma' \vdash \tyeq \alpha \beta$ are
$\Gamma_\sep := \alpha : \sep, \beta : \sep$, and
$\Gamma_\deepsep := \alpha : \deepsep, \beta : \deepsep$. To check
that $\guard {\tyeq \alpha \beta} \kappa : m$ holds in $\Gamma$, our
rule asks us to check that $\kappa : m$ holds in both $\Gamma_\sep$
and $\Gamma_\deepsep$. But in fact we, the implementers of the checking
algorithm, know that making stronger assumptions in the context makes
more mode-checks pass, so it suffices to check in context
$\Gamma_\sep$.\footnote{We could formulate the rule to only check the unique
minimal context $\Gamma'$, but Fact~\ref{fact:non-principal} in
Section~\ref{sec:semantics} suggests that such a unique context may
not always exist.}

Finally, it is interesting to consider the derivation of the following
judgment, which represents in our system the key ingredient of the
\caml{_ first} GADT that led to the introduction of $\deepsep$ in
Section~\ref{subsec:intro-deepsep}.
\[
  \alpha : m \vdash
    \exists \beta \gamma.\, \guard {\tyeq \alpha {\beta \times \gamma}} \beta
    : \sep
\]
We expect to accept this type declaration only in the case where $\alpha :
\deepsep$, this assumption guaranteeing that $\beta$ will be separable.
Opening the existentials puts $\beta, \gamma$ in the context at $\ind$:
\[
  \alpha : m, \beta : \ind, \gamma : \ind \vdash
    \guard {\tyeq \alpha {\beta \times \gamma}} \beta
    : \sep
\]
Let us now reason by case analysis on $m$ to show that only
$m = \deepsep$ has a valid derivation of this judgment.
\begin{itemize}
\item If $m$ is $\deepsep$, then any $\Gamma' \geq \Gamma$ with
  $\Gamma' \vdash \tyeq \alpha {\beta \times \gamma}$ has
  $\beta:\deepsep, \gamma:\deepsep$ since otherwise
  $\beta \times \gamma : \deepsep$ cannot hold. In particular,
  $\Gamma' \vdash \beta : \sep$ holds as expected, so the judgment is
  derivable.
\item If $m$ is $\sep$ or $\ind$, then
  $\Gamma' := \alpha:\sep, \beta:\ind, \gamma:\ind$ satisfies
  $\Gamma' \geq \Gamma$ and
  $\Gamma' \vdash \tyeq \alpha {\beta \times \gamma}$, but we do not
  have $\Gamma' \vdash \beta : \sep$: the judgment is not derivable.
\end{itemize}

\paragraph{Datatypes.} The judgment for datatypes is simple: a datatype
is accepted if its set of values is separable. Boxed
variant or record definitions are always separable, so no premises
are required. Unboxed variants/records with parameter type $\tau$,
or type synonyms for $\tau$ require $\tau : \sep$.

\paragraph{Definition blocks.} The judgment for definition blocks has
an input signature, $\Sigma_{\mathsf{env}}$ in the rule, which lists
assumptions that we make on the datatype constructors provided by the
typing environment, and an output signature, $\Sigma_{\mathsf{block}}$
in the rule, which is a valid signature for the current block of
mutually-recursive definitions. A definition block is valid if each
datatype definition it contains is valid. Note that each definition
is checked with the full $\Sigma_{\mathsf{block}}$ in the signature
context: all mutually-recursive datatype constructors are available in
scope.

\paragraph{Meta-theory.}

The following two lemmas can be easily proved by structural induction
on typing derivations.

% Gabriel: the Cut elimination lemma can be removed if we lack space
% Rodolphe: it is a lot shorter without the alternative.
\begin{lemma}[Cut elimination]\label{lem:cut-elimination}
  A separability derivation can always be rewritten so that all
  occurrences of the conversion rule only have axiom rules as their
  premises
  or, equivalently,
  using the following more primitive rule.
  \begin{mathpar}
   \infer
   {(\alpha : m) \in \Gamma \\ m \geq n}
   {\Sigma; \Gamma \vdash \alpha : n}
  \end{mathpar}
\end{lemma}

\begin{lemma}[Monotonicity]\label{lem:monotonicity}
  The following rule is admissible.
  \begin{mathpar}
    %\mprset{fraction={{}{~\cdots}{}}} % FIXME why the dots?
    \infer
    {\Sigma \leq \Sigma' \\
     \Gamma \leq \Gamma' \\
     \Sigma; \Gamma \vdash \tau : m \\
     m \geq m'}
   {\Sigma'; \Gamma' \vdash \tau : m'}
  \end{mathpar}
\end{lemma}

\section{Semantics}
\label{sec:semantics}

Due to space limitations, we moved our presentation of separability
semantics to Appendix~\ref{app:semantics}. It contains a precise
semantic characterization of separability judgments in the flavor of
set-theoretic or realizability models, and a discussion of soundness,
principality and completeness. Many results are left as conjectures --
we explain why some of them are surprisingly delicate.

\section{Integration into OCaml}%%%%%%%%%%%%%%%%%%%%%%%%%%%%%%%%%%%%%%%%%%%%%%%%

We are now going to highlight some important points of our implementation
in the OCaml type-checker. The corresponding code has been proposed for
integration through a GitHub pull request, that is visible at the following
URL.
\begin{center}
  \url{https://github.com/ocaml/ocaml/pull/2188}
\end{center}

Our implementation is derived from the type system given in
Section~\ref{sec:separability_inference} by inferring mode signatures
for type definitions. It only assigns mode signatures that can be
justified by our syntactic type system. For example,
if \caml{('a : Sep, 'b : Ind) t} is assigned to a type
definition \caml{type ('a, 'b) t = A}, then
$
 \Sigma_{\mathsf{env}} \vdash
 t(\alpha, \beta) := A \dashv t(\alpha : \sep, \beta : \ind)
$
must be derivable using the inference rules of Figure~\ref{fig:rules}.

\subsection{Inferring block signatures with a fixpoint}

Our main checking function constructs a block signature
$\Sigma_{\mathsf{block}}$ given an environment $\Sigma_{\mathsf{env}}$
and a block of type definitions
$(t_i(\vec{\alpha_i}) := A_i)_{1 \leq i \leq n}$:
\begin{minted}{ocaml}
val check : Env.t -> type_definition ConstrMap.t -> mode_signature ConstrMap.t
\end{minted}
In this signature, \caml{type_definition ConstrMap.t} maps datatype
constructors $t_i(\vec{\alpha})$ to datatype definitions $A_i$, and
\caml{mode_signature ConstrMap.t} them to a mode signature
$t_i(\vec{\alpha:m})$.

It is not possible to directly compute mode signatures for mode
definitions, due to recursive types: we need to know the mode
signature of the type constructors in order to assign them a mode
signature. This circularity is solved by using a fixpoint computation:
we iteratively refine an approximation of the mode signatures.

The fixpoint computation starts with the most permissive mode
signature, which only requires mode $\ind$ for the variables of every
type constructor of the block. At each iteration, the separability of
every type of the block is checked against the current approximation
of the mode signatures, accumulating more precise constraints whenever
required. If the mode signatures coming from these constraints are more
demanding than those of the current approximation, we define them as the new
approximation and continue. Otherwise, we have found a mode signature
that validates the judgment -- it is sound. In fact, it is the most
permissive mode signature that we can find by iteration in this way.

\subsection{Management of GADTs}
\label{subsec:implem-gadts}

As explained in Section~\ref{subsec:gadts-reminder}, our type system
for separability does not directly account for GADTs: they are encoded
using existential quantifiers and equality guards. Our implementation
handles only GADTs, which correspond to a very specific mode of use of
existentials and guards within type declarations.

Consider, for example, \caml{type 'a fun = Fun : ('b -> 'c) ->
  ('b -> 'c) fun}. The return type \caml{('b -> 'c) fun} determines
the equality guard $\tyeq{\mathtt{'a}}{\mathtt{'b} \to \mathtt{'c}}$,
and the existential types are the free type variables $\mathtt{'b}$,
$\mathtt{'c}$ of the declaration. In general, unboxable GADT types contain
existential quantifiers only at the toplevel, immediately followed by
one equation for each type parameter, with finally a concrete parameter
type $\tau$.

The first thing to note is that all existential quantifiers occurs
at the top level: they exactly correspond to the free
variables of the parameter type. We can infer a mode signature for the
parameter type $\tau$ and, following our rule for existentials, check
that the modes inferred for the free existential variables in $\tau$
are $\ind$ and fail otherwise. Finally, the mode signature for the
definition of the GADT is obtained by removing the existential
variables from the inferred mode signature.

The delicate matter with in handling of GADTs is unsurprisingly the
management of equality guards. Recall that an unboxed GADT
$\mathtt{type}~t(\vec{\alpha}) = \mathtt{K} : \tau \to t(\vec{\kappa})$
is encoded as
$\exists \vec{\beta}.  \guard{(\vec{\alpha = \kappa})}{\tau}$,
where the $\vec{\beta}$ are the type variables free in
$\tau, \vec{\kappa}$ and the $\vec{\kappa}$ do not contain any
$\alpha_i$. The idea of the implementation is to first infer a mode context
$\Gamma$ such that $\Sigma; \Gamma \vdash \tau : \sep$, assigning
modes to the variables in $\vec{\alpha}$ and $\vec{\beta}$. Equations
are then discharged one by one, refining $\Gamma$ in the process,
before the existential variables $\vec{\beta}$ are checked to be
$\ind$ and eliminated to create the final context for the GADT
parameters only -- the mode signature. Every equation, taken in any
order, is managed according to the following three cases:
\begin{itemize}
  \item An equation of the form $(\alpha_i = \beta_j)$ leads to $\Gamma$ being
    updated with $\{\alpha_i \mapsto \Gamma(\beta_j), \beta_j \mapsto \ind\}$.
  \item An equation of the form $(\alpha_i = \kappa)$ with
    $\Gamma \cap FV(\kappa)$ only containing $\ind$ leads to the
    equation being discarded without any update to $\Gamma$.
  \item Any other equation (i.e., equations $(\alpha_i = \kappa)$ with
    $\Gamma \cap FV(\kappa)$ not only containing $\ind$) leads
    to $\Gamma$ being updated with
    $\{\alpha_i \mapsto \deepsep, FV(\kappa) \mapsto \ind \}$.
\end{itemize}
These transformations respect the following invariant: if the equation
$\tyeq \alpha \kappa$ updates $\Gamma$ into $\Gamma_0$, then any
$\Gamma' \geq \Gamma_0$ that satisfies $\Gamma' \vdash \tyeq \alpha \kappa$
is above the original $\Gamma$ ($\Gamma' \geq \Gamma$). In other words,
strengthening the resulting context $\Gamma_0$ with the equation we just handled
would give a context as permissive or more than the original context $\Gamma$.

\subsection{Cyclic types}
\label{subsec:cyclic_types}

The OCaml type system accepts equi-recursive types. By default, any
cycle in types must go through a polymorphic variant or object type
constructor, but the \caml{-rectypes} option generalizes this to any
ground type constructor. A type such as \caml{('a -> 'b) as 'b}, for example,
gives a cyclic representation to the infinite type \caml{'a -> 'a -> 'a -> ...},
and must be supported by our implementation.

To support cyclic types, we extend our checking rules to support
a form of coinduction: each time we reduce a judgment to simpler
premises (for example from $\Gamma \vdash \tau_1 \to \tau_2 : m$ to
$\Gamma \vdash \tau_i : \mcomp m \ind$), we record in the premises
that we have previously encountered the judgment
$\tau_1 \to \tau_2 : m$. If, later, we encounter the same judgment to
prove, we terminate immediately with a success.

With this rule, it is possible to prove both
\begin{mathpar}
  \alpha : \ind \vdash ((\alpha \to \beta) \mathtt{~as~} \beta) : \sep

  \alpha : \deepsep \vdash ((\alpha \to \beta) \mathtt{~as~} \beta) : \deepsep
\end{mathpar}
Informally, the reason why recursive occurrences of the same judgment
can be considered an immediate success is that we ``made progress''
between the first occurrence of the cyclic type $\beta$ and its second
occurrence in $\alpha \to \beta$: the second occurrence is ``guarded''
under a value constructor, and the set of values of $\beta$ we are
classifying as separable, or deeply separable, is not the one we
started from -- that would be an invalid cyclic reasoning -- but a copy
of it occurring deeper in the type structure.

However, some cyclic types such as $t(\alpha) \mathtt{~as~} \alpha$
have a less clear status, as the recursive occurrence is not guarded
under a computational type constructor (arrow, product...), but under
an abbreviation. Is this type well-defined if, for example,
$t(\alpha)$ is defined as $t(\alpha) := \alpha$? OCaml rejects some of
these dubious-looking circular types, but instead of trusting our fate
to the rest of the type-checker we decided to account for them in our
theory. We split our set of co-inductive hypotheses (the list of
judgments that we are trying to prove) into a set of ``safe''
hypotheses $\Theta_{\mathsf{safe}}$, which can be used immediately,
and a set of ``unsafe'' hypotheses, which can only be used after
a computation type constructor has been traversed.

Here are some rules of the corresponding formal system, extended with
coinductive hypotheses, which guided our OCaml implementation:
\begin{mathpar}
  \infer
  {\forall i \in \{1, 2\},\quad
    \Sigma; \Gamma; \Theta_{\mathsf{safe}}, \Theta_{\mathsf{unsafe}},
    (\tau_1 \to \tau_2 : m); \emptyset \vdash \tau_i : \mcomp m \ind
  }
  {\Sigma; \Gamma; \Theta_{\mathsf{safe}}; \Theta_{\mathsf{unsafe}}
    \vdash \tau_1 \to \tau_2 : m}

  \infer{
    t(\alpha_1 : m_1, \dots, \alpha_n : m_n) \in \Sigma
    \\
    (\Sigma; \Gamma; \Theta_{\mathsf{safe}};
     \Theta_{\mathsf{unsafe}}, (t(\tau_1, \dots, \tau_n) : m)
     \vdash \tau_i : \mcomp m {m_i})_{1 \leq i \leq n}
  }{\Sigma; \Gamma; \Theta_{\mathsf{safe}}; \Theta_{\mathsf{unsafe}}
    \vdash t(\tau_1, \dots, \tau_n) : m}

  \infer
  {(\tau : m') \in \Theta_{\mathsf{safe}} \\ m' \geq m}
  {\Sigma; \Gamma; \Theta_{\mathsf{safe}}; \Theta_{\mathsf{unsafe}} \vdash \tau : m}

  \infer
  {(\tau : m) \in \Theta_{\mathsf{unsafe}}}
  {\Sigma; \Gamma; \Theta_{\mathsf{safe}}; \Theta_{\mathsf{unsafe}} \nvdash \tau : m}
\end{mathpar}
The arrow rule has its premises under a computational type
constructor, so it passes all coinductive hypotheses, including the
new assumption on $\tau_1 \to \tau_2$, to its safe set. A datatype
constructor may be just an abbreviation, so it adds the new hypothesis
to the unsafe set.

Finally, whenever a judgment needs to be proved, it immediately
succeeds if a stronger hypothesis is in the safe set. On the other
hand, if our hypothesis is already in the unsafe set, then we know
that it cannot be proven without recursively assuming itself, and we
can in fact fail with an error.

\subsection{Non-conservativity}

Section~\ref{subsec:principality}, Fact~\ref{fact:non-principal} shows
that our modes do not admit principal judgments -- in particular our
judgments are non-principal. In particular, the following OCaml
declaration can be given either signatures
$t(\alpha:\ind, \beta:\sep)$ and $t(\alpha:\sep, \beta:\ind)$.
\begin{minted}{ocaml}
  type ('a, 'b) t = K : 'c -> ('c, 'c) t
\end{minted}

Our implementation will choose one of the two minimum signatures,
depending on the order in which constraints are handled -- see
Section~\ref{subsec:implem-gadts}. This means that some correct uses
of the type will be disallowed: no matter which one is chosen, one of
the two following declarations will be rejected:
\begin{minted}{ocaml}
  type t1 = T1 : (int, 'a) t -> t1
  type t2 = T2 : ('a, int) t -> t2
\end{minted}
On the other hand, the current implementation, which expands the
definition of \caml{t}, accepts both definitions.

We have decided to accept these completeness regressions. Our
implementation is cleaner, accepts important examples rejected by the
current implementation, and is safer in its handling of cycles,
without fuel. In contrast, the counter-examples we could build are
fairly esoteric, and we have not found any of them in the current
testsuite or user programs. Only time will tell if this assumption is
reasonable; we discuss ideas to recover principality (and accept both
\caml{t1} and \caml{t2}) in Future Work Section~\ref{par:richer-modes}.

\section{Related and future work}%%%%%%%%%%%%%%%%%%%%%%%%%%%%%%%%%%%%%%%%%%%%%

\subsection{Future work}

\paragraph{Richer modes}
\label{par:richer-modes}

The implementation of our unboxability check is satisfactory in the
sense that it accepts most valid (unboxed) definitions. There is
however room for improvement, especially in the handling of type
equality guards, be they in GADTs or in toplevel equality
constraints.

For instance, the way we handle equations of GADTS (see the
previous section) is in some sense incomplete, and our language of
modes is not principal. We considered extending our modes with modes
of the form $\exists (\beta : m).\, m'$ and
$\guard {\tyeq \alpha \tau} m$. However, this would significantly
increase the complexity of the theory and the implementation, only to
handle corner cases that may not be worth it.

Another simpler approach to regain an impression of principality
would be to support disjunctions of modes. In the problematic example
$t(\alpha, \beta) := \guard{\tyeq \alpha \beta} \beta$, there are two
minimum modes $t(\alpha:\ind, \beta:\sep)$ and
$t(\alpha:\sep, \beta:\ind)$, so this type could be given the
principal mode
$t (\alpha:\ind, \beta:\sep)\vee(\alpha:\sep,\beta:\ind)$.

\paragraph{Automatic unboxing}

Currently, unboxable type declarations are never unboxed
automatically, the user has to explicitly ask for it. Automatic
unboxing has been considered, but it could break existing code using
the foreign function interface. For example, a C function receiving an
OCaml value from an unboxable but non-unboxed type currently needs to
unbox it explicitly. The same action on an (automatically) unboxed
type would fail if the C code is not changed accordingly.

\paragraph{Disjoint GADT unboxing}

One could wish to unbox multi-constructors GADTs in the case where
typing constraints lead to the mutual exclusion of the different
constructors.
\begin{minted}{ocaml}
  type _ value =
    | Int : int -> int value
    | Bool : bool -> bool value
    | Pair : ('b * 'c) -> ('b * 'c) value
\end{minted}
While each ground instance of this value type has a single possible
constructor and could be unboxed, pattern-matching on an \caml{'a value}
would then have to be disallowed: pattern-matching learns the value of
\caml{'a} from inspecting the GADT constructor, which is not present
anymore in the unboxed representation. It seems fishy to only
allow pattern-matching on partially-determined instances of the type,
and we did not investigate further.

\subsection{Just get rid of the damn float thing}

The fairly elaborate sufferings we just went through are caused by the
dynamic unboxing optimization for arrays of floating point numbers. If
this dynamic optimization were removed, we would not need separability
anymore and the unboxing check could also be removed. This dynamic
check also has consequences on other features: by making a $\top$ type
(our $\exists \alpha.\,\alpha$) unsound, it prevents extending the
relaxed value restriction to generalize contravariant
variables~\citep[page 11]{garrigues04}.

There is an ongoing debate in the OCaml community on this idea. An
experimental configure flag \texttt{-no-flat-float-array} can be set
to disable dynamic flat representation optimizations in the
implementation, and benefit from the simpler type theory. Since 4.06
(November 2017), a new primitive monomorphic type \texttt{floatarray}
exists that is specialized for unboxed float arrays, and can be used
by users intending to use this optimization, but it lacks library
support and convenient array-indexing notations. The problem is with
generic code, written against parametric \texttt{'a array} value and
then applied in numeric programs on float array, whose performance
would be silently degraded with an important slowdown. In other terms,
removing the dynamic optimization would be acceptable for experts
authors of numerical code willing to modify their codebase, but hurt
the performance of programs written naively by beginners.

\subsection{Related work}

We discussed the existing implementation of the unboxability check in
Section~\ref{subsec:existing-check}.

The approach presented in our work is largely inspired from the way
the variance of type declarations is handled in languages with
subtyping \citep*{abel08, scherer13}.

The memory representation of values used in OCaml and similar
languages finds its origin in Lisp-like languages
\citep*{leroy90}. Despite the advantage of allowing every value to
have the same, single-word representation, this approach also has
obvious limitations in terms of performances due to the introduction
of indirections. As a consequence, ways of lowering this overhead in
certain scenarios have been investigated, one possibility being to mix
tagged and untagged representations \citep*{peterson89, leroy90,
  leroy92}. Another idea that has been investigated is to consider
unboxed values as first-class citizens, although distinguished by
their types \citep*{peyton-jones91}.

%\newpage
\label{sect:bib}
\bibliography{biblio}

\newpage
\appendix

\section{Semantics}
\label{app:semantics}

\newcommand{\GType}{\mathsf{GType}}
\newcommand{\ground}[1]{\underline{#1}}
\newcommand{\sem}[1]{\llbracket #1 \rrbracket}

\newcommand{\sat}[3]{#2 \vDash^{#1} #3}
\newcommand{\unsat}[3]{#2 \nvDash^{#1} #3}
\newcommand{\satdecl}[3]{#2 \vDash^{#1}_{\mathsf{decl}} #3}
\newcommand{\satsig}[2]{#1 \vDash #2}

\begin{figure}[htbp]
Ground/closed values $\ground{v}$,
type expressions $\ground{\tau}$,
datatypes $\ground{A}$,\\
context valuations $\ground{\gamma}$,
datatype signature valuations $\sigma$
\begin{mathpar}
  \begin{array}{l@{~}r@{~}l@{\quad}l}
    \ground{v}
    & \bnfeq & \builtin{true}, \builtin{false} & \mbox{booleans} \\
    & \bnfor & \operatorname{int}(n \in \mathbb{N}) & \mbox{integers} \\
    & \bnfor & \operatorname{float}(x \in \mathbb{R}) & \mbox{float} \\
    & \bnfor & (\ground{v_1}, \dots, \ground{v_n}) & \mbox{tuple} \\
    & \bnfor & \builtin{function} & \mbox{function} \\
    & \bnfor & \{ l_1 : \ground{v_1}; \dots; l_n : \ground{v_n} \}
                                               & \mbox{record} \\
    & \bnfor & C_i~\ground{v} & \mbox{variant} \\
  \end{array}
  \quad
  \begin{array}{rll}
    \operatorname{closed}(\tau)
    & :=
    & \forall \alpha,\ \alpha \notin \tau
    \\
    \GType
    & :=
    & \{ \ground{\tau} \mid \operatorname{closed}(\tau) \}
    \\[1em]
    \operatorname{closed}(A)
    & :=
    & \forall \tau \in A,\ \operatorname{closed}(\tau)
    \\
    \mathsf{GDatatype}
    & :=
    & \{ \ground{A} \mid \operatorname{closed}(A) \}
    \\[1em]

    \ground{\gamma}
    & \in
    & \mathcal{P}(\mathsf{TypeVar} \to \GType)
    \\
    \sigma
    & \in
    & \mathcal{P}(\mathsf{TypeConstructor} \to \mathsf{GDatatype})
  \end{array}
\end{mathpar}

\medskip

Values at a ground type expression $\sat \sigma {\ground v} {\ground \tau}$
\begin{mathpar}[\mprset{lineskip=.8em}]
  \infer{ }{\sat \sigma {\builtin{true}, \builtin{false}} {\ground {\builtin{bool}}}}

  \infer{ }{\sat \sigma {\operatorname{int}(n)} {\ground{\builtin{int}}}}

  \infer{ }{\sat \sigma {\operatorname{float}(x)} {\ground{\builtin{float}}}}

  \infer{ }
  {\sat \sigma {\builtin{function}} {\ground{\tau_1 \to \tau_2}}}

  \infer{\left(\sat \sigma {\ground{v_i}} {\ground{\tau_i}}\right)_{1 \leq i \leq n}}
  {\sat \sigma {(\ground{v_1}, \dots, \ground{v_n})}
        {\ground{(\tau_1 \times \dots \times \tau_n)}}}

  \infer{\exists \ground{\tau} \in \GType,\quad
    \sat \sigma {\ground v} {\ground{\kappa[\ground{\tau}/\alpha]}}}
  {\sat \sigma {\ground v} {\ground{\exists \alpha. \kappa}}}

  \infer{\forall \ground{\tau} \in \GType,\quad
    \sat \sigma {\ground v} {\ground{\kappa[\ground{\tau}/\alpha]}}}
  {\sat \sigma {\ground v} {\ground{\forall \alpha. \kappa}}}

  \infer{
    (t(\vec{\alpha}) := A) \in \sigma \\
    \satdecl \sigma {\ground v}
      {\ground{A[\ground{\tau_1} / \alpha_1, \dots, \ground{\tau_n} / \alpha_n]}}}
  {\sat \sigma {\ground v} {\ground {t(\tau_1, \dots, \tau_n)}}}

  \infer{
    (\ground{\tau_1} = \ground{\tau_2})
    \implies
    {\sat \sigma {\ground v} {\ground \kappa}}}
  {\sat \sigma {\ground v} {\ground {\guard {\tyeq {\tau_1} {\tau_2}} \kappa}}}
\end{mathpar}

\medskip

Values at a ground datatype $\satdecl \sigma {\ground v} {\ground A}$
\begin{mathpar}[\mprset{lineskip=.8em}]
  \infer
  {\sat \sigma {\ground v} {\ground{\tau_i}}\\ 1 \leq i \leq n}
  {\satdecl \sigma {C_i~\ground{v}}
    {\ground{C_1 \of \tau_1 \mid \dots \mid C_n \of \tau_1}}}

  \infer
  {\sat \sigma {\ground v} {\ground \tau}}
  {\satdecl \sigma {\ground v} {\ground{C \of \tau \; \mathtt{[@@unboxed]}}}}

  \infer
  {\left( \sat \sigma {\ground{v_i}} {\ground{\tau_i}} \right)_{1 \leq i \leq n}}
  {\satdecl \sigma {\{ l_1 : \ground{v_1}; \dots; l_n : \ground{v_n} \}}
    {\ground{\{ l_1 : \tau_1; \dots; l_n : \tau_n \}}}}

  \infer
  {\sat \sigma {\ground v} {\ground \tau}}
  {\satdecl \sigma {\ground v} {\ground{\{ l : \tau \}  \; \mathtt{[@@unboxed]}}}}

  \infer
  {\sat \sigma {\ground v} {\ground \tau}}
  {\satdecl \sigma {\ground v} {\ground \tau}}
\end{mathpar}

\medskip

Ground types at a mode $\sat \sigma {\ground \tau} m$,
ground valuations at a context $\sat \sigma {\ground \gamma} \Gamma$
\begin{mathpar}[\mprset{lineskip=.8em}]
  \begin{array}{lll}
    \operatorname{isfloat}(\ground{v})
    & :=
    & \exists x,\ \ground{v} = \operatorname{float}(x)
    \\
    \operatorname{separable}(X)
    & :=
    & (\forall \ground{v} \in X,\ \operatorname{isfloat}(v))
      \vee
      (\forall \ground{v} \in X,\ \neg\operatorname{isfloat}(v))
  \end{array}
  \\
  \infer{ }
  {\sat \sigma {\ground \tau} \ind}

  \infer{\operatorname{separable}(\{
      \ground v \mid \sat \sigma {\ground v} {\ground \tau} \}
    \})}
  {\sat \sigma {\ground \tau} \sep}

  \infer{\forall {\sub {\tau'} \tau},\quad
    \sat \sigma {\ground {\tau'}} \sep}
  {\sat \sigma {\ground \tau} \deepsep}

  \infer
  {\forall (\alpha : m) \in \Gamma,\quad
   \sat \sigma {\ground{\gamma(\alpha)}} m}
  {\sat \sigma {\ground \gamma} \Gamma}
\end{mathpar}

Remark: we have $\unsat \sigma {\ground{\exists \alpha.\,\alpha}} \sep$,
which means that not all ground types are separable.

\medskip

Ground datatypes at a mode $\satdecl \sigma {\ground A} m$,\\
parametrized datatypes at a mode signature $\sat \sigma A {t(\vec{\alpha:m})}$,\\
datatype definitions at a block signature $\sat \sigma \sigma \Sigma$
\belowdisplayskip 0em
\begin{mathpar}[\mprset{lineskip=.8em}]
  \infer
  {\operatorname{separable}(\{ \ground v \mid \satdecl \sigma {\ground v} {\ground A} \})}
  {\satdecl \sigma {\ground A} \sep}

  \infer{\forall {\sat \sigma {\ground \gamma} {\vec{\alpha:m}}},\quad
    \satdecl \sigma {\ground{\gamma(A)}} \sep}
  {\sat \sigma A {t(\vec{\alpha:m})}}

  \infer{\forall (t(\vec{\alpha}) := A) \in \sigma,\quad
    \sat \sigma A {\Sigma(t)}}
  {\satsig {\sigma} \Sigma}
\end{mathpar}
\caption{Ground semantics}
\label{fig:ground-semantics}
\end{figure}

\begin{figure}[htbp]
  \abovedisplayskip 0em
  \belowdisplayskip 0em
  \begin{mathpar}[\mprset{lineskip=.8em}]
    \infer
    {\forall {\satsig \sigma \Sigma},
     \forall {\sat \sigma \gamma \Gamma},
     \quad
     \sat \sigma {\gamma(\tau)} m}
    {\Sigma; \Gamma \vDash \tau : m}

    \infer
    {\forall {\satsig \sigma \Sigma},
     \forall {\sat \sigma \gamma \Gamma},
     \quad
     \satdecl \sigma {\gamma(A)}}
    {\Sigma; \Gamma \vDash_{\mathsf{decl}} A}

    \infer
    {\forall {\satsig {\sigma_0} {\Sigma_0}},
     \quad
     \satsig {\sigma_0, \sigma} {\Sigma_0, \Sigma}
     }
    {\Sigma_0 \vDash \sigma \Dashv \Sigma}
  \end{mathpar}
  \caption{Judgment semantics}
  \label{fig:judgment-semantics}
\end{figure}

We give a semantic model of types, datatypes and modes in
Figure~\ref{fig:ground-semantics}. The general idea is to be able to
interpret the judgment $\Gamma \vdash \tau : m$ as follows: ``if we
replace each type variable $(\alpha:m_\alpha)$ of $\Gamma$ by
a closed/ground type with the correct separability mode $m_\alpha$,
then $\tau$ will really have the separability mode $m$''.

\paragraph{Ode to semantics.}

In Section~\ref{sec:separability_inference} we have given a set of
inference rules to prove judgments of the form
$\Sigma; \Gamma \vdash \tau : m$, guided by the intuition that this
\emph{should} provide evidence that $\tau$ is separable. Inference
rules are our key contribution as they can easily be turned into
a checking or inference algorithm, but they are also subtle, and may
very well be wrong. They are hard to audit by someone else who does
not trust us, the authors: there are a lot of details to check.

The point of a semantics is to provide an ``obvious'' counterpart to
inference rules. A formal definition of separability, or whatever
property an inference system is trying to capture, that is
self-evident, can be easily checked and trusted by other people -- in
particular, it does not depend on the previous inference rules in any
way. It serves as a specification, can use arbitrary mathematical
operations, and does not need to be computable or close to an
algorithm. Finally, one wishes to prove that the syntactic inference
rules and the semantics coincide in some sense; this implies that we
can trust the inference rules as much as we trust the semantics.

Our semantics can also be understood as an idealized ``model'' of the
programming language we are studying: it contains the assumptions that
we make about the language and type system, and
can be compared to OCaml to ensure that those assumptions are correct.

\paragraph{Ground value semantics.}

To define what it means for a closed type $\ground{\tau}$ to have the
mode $m$, we use the set-theoretic intuition introduced earlier:
``its values are all floating-point numbers, or none of them are''. To
do so formally, we introduce a syntax of closed/ground values, and
specify which ground values inhabit which ground types.

Our notion of ground/closed values is idealized; in particular, we
represent all values at a function type as an opaque
$\builtin{function}$ blob, as if all functions were represented in the
same way in the programming languages we model. They are not, but the
differences in representation are irrelevant to reason about
separability.

In the figure we define ground values $\ground{v}$ (just data, no
term variables), ground types $\ground{\tau}$ (no free
type variables), ground datatypes (no free type variables
or parameters), and blocks of datatype definitions $\sigma$, which
assign a datatype to each datatype constructor of a signature
$\Sigma$.

Finally, we define a series of semantic judgments using the symbol
$(\vDash)$. The judgment $\sat \sigma {\ground v} {\ground \tau}$
specifies when a value inhabits a type, relative to a block of
definitions $\sigma$ to interpret type constructors.  The judgment
$\sat \sigma {\ground v} {\ground A}$ specifies when a value inhabits
a type parameter; note in particular how the judgments for
\texttt{unboxed} datatypes reflect what happens in an implementation.

The judgments $\sat \sigma {\ground \tau} m$ and
$\sat \sigma {\ground A} m$ specify when a type expression or
datatype respect the separability mode $m$. They are defined in terms
of our set-theoretic characterizations,
$\operatorname{separable}(X)$. Finally,
$\sat \sigma A t(\vec{\alpha:m})$ specifies when a parametrized
datatype $A$ respects a mode signature, and $\satsig \sigma \Sigma$
specifies that a definition block respects a signature block.

In Figure~\ref{fig:judgment-semantics}, we use these specifications to
build a semantic counterpart for each of our syntactic judgments.  This
lets us easily formulate soundness and partial completeness
results. For example, $\Sigma; \Gamma \vDash \tau : m$ is the semantic
counterpart of the judgment $\Sigma; \Gamma \vdash \tau : m$. It
captures what it means for the judgment to hold: for any valid
definitions $\satsig \sigma \Sigma$, and any choice of ground types
$\ground{\gamma}$ valid for $\Gamma$
($\sat \sigma {\ground \gamma} \Gamma$), replacing the variables in
$\tau$ by the ground types in $\ground{\gamma}$ gives a ground type
$\ground{\gamma(\tau)}$ at mode $m$, that is,
$\sat \sigma {\ground{\gamma(\tau)}}$ holds.

One does not have to trust the inference rules of our syntactic judgment
$\Sigma; \Gamma \vdash \tau : m$, which may very well be wrong, to
trust that this semantic judgment $\Sigma; \Gamma \vDash \tau : m$
captures a good notion of separability. One need to read the
declarative rules of figures~\ref{fig:ground-semantics}
and~\ref{fig:judgment-semantics}, which we designed to be obvious. It
is now evident what the right statements should be for soundness and
completeness of our syntactic inference system: it is sound if the
syntactic judgment implies the semantic judgment (all judgments with
a syntactic derivation are semantically true), and complete if the
converse holds (all true judgments can be established by
a syntactic derivation).

\subsection{(Maybe-)Soundness and (in-)completeness}

In a research paper, the perfect way to evaluate a syntactic system of
inference rules is to provide an independent declarative semantics for
it, state the corresponding soundness/completeness results, and prove them.

In real-world implementations of a programming language, the syntactic
system is rarely written down (often, only the algorithm inspired from
the rules is kept), the declarative semantics are only vague
intuitions in the mind of the rule designers, and the statements are
never formulated precisely enough to dream of a proof.

In this imperfect research paper, we have precise rules, and an
independent declarative semantics, and precise statements... but most
proofs are missing. In fact, it is fairly non-obvious that those
proofs exist: as we were doing this precise formalization work we have
discovered that completeness does not hold in presence of constraints,
and even that the system is not principal -- not only the inference
rules, but even the semantics.

This subsection contains the best state of our knowledge and hopes
about the formal status of our inference rules.

\begin{conjectured_lemma}[Type expression soundness]
  $
    \Sigma; \Gamma \vdash \tau : m
    \;\implies\;
    \Sigma; \Gamma \vDash \tau : m
  $
\end{conjectured_lemma}

Remark: it is not obvious that the existential rule is sound, for
example. It proves $\exists \alpha.\,\tau : m$ under a premise
$\alpha : \ind \vdash \tau : m$. The values of $\exists \alpha.\,\tau$
are the union of all the $\ground{\tau[\ground \sigma / \alpha]}$, and
our induction hypothesis tells us that each
$\ground{\tau[\ground \sigma / \alpha]}$ indeed has the mode $m$. But
separability is not at all stable by union: both $\builtin{int}$ and
$\builtin{float}$ are separable, but their union is not. One needs to
argue that either all $\ground{\tau[\ground \sigma / \alpha]}$ are
inhabited by floating-point numbers only, or that none of them contain
floating-point numbers, and this is a subtle property of the structure
of algebraic datatypes.

\begin{conjectured_theorem}[Soundness]
  $
    \Sigma_0 \vdash \sigma \dashv \Sigma
    \;\implies\;
    \Sigma_0 \vDash \sigma \Dashv \Sigma
  $
\end{conjectured_theorem}

\newcommand{\CF}[1]{\mathsf{CF}(#1)}

\begin{fact}[Incompleteness]
  \label{fact:incompleteness}
  There exists a $\tau$ such that $\emptyset; \emptyset \vDash \tau : m$ but
  $\emptyset; \emptyset \nvdash \tau : m$.
\end{fact}
\begin{proof}
  Consider the two types
  $
    \tau_1
    \;:=\;
    \exists \beta. \guard {\tyeq {\builtin{int}} {\builtin{bool}}} \beta
  $ and $
    \tau_2
    \;:=\;
    \exists \alpha \beta.\,\guard {\tyeq \alpha {\alpha \times \builtin{int}}} \beta
  $.
  They are semantically separable, but our type system cannot prove
  it. The reason why they are semantically separable is that the equations
  never hold in our semantic model with finite ground types, so the
  constrained type is empty and thus trivially separable.

  We could make our syntactic inference rules more complete by adding
  more equality-reasoning power to the judgment
  $\Sigma; \Gamma \vdash \tyeq {\tau_1} {\tau_2}$. For the first case
  $\tyeq {\builtin{int}} {\builtin{bool}}$, incompatible type
  constructors should not be considered equal. In the second case,
  $\tyeq \alpha {\alpha \times \builtin{int}}$, one could add an
  occurs-check to rule out such recursive types, but note that this
  property, true in our model, may not hold in the real world -- it
  does not in OCaml when \texttt{-rectypes} is used.
\end{proof}

This result suggests that equality-reasoning can be fairly subtle and
that we should not expect syntactic completeness on types with
constraints. We can still hope to have completeness in their absence.

\begin{definition}[Constraint-free]
  A type $\tau$, datatype $A$ or definition block $\sigma$ is
  \emph{constraint-free} if it does not contain any type equality
  constraint ($\guard{\tyeq \alpha {\tau'}} \kappa$). We write
  $\CF{\tau}$, $\CF{A}$ or $\CF{\sigma}$.
\end{definition}

\begin{conjectured_lemma}[Completeness on constraint-free type expressions]
\[
  \CF{\tau}
  \wedge
  \Sigma; \Gamma \vDash \tau : m
  \;\implies\;
  \Sigma; \Gamma \vdash \tau : m
\]
\end{conjectured_lemma}

\begin{conjectured_theorem}[Completeness on constraint-free signatures]
\[
\CF{\sigma}
\wedge
\Sigma_0 \vdash \sigma \dashv \Sigma
\;\implies\;
\Sigma_0 \vDash \sigma \Dashv \Sigma
\]
\end{conjectured_theorem}

\subsection{Constraint-free principality}
\label{subsec:principality}

Unfortunately, there is more bad news to come for equality constraints.

\begin{fact}[Semantic non-principality]
  \label{fact:non-principal}
  There exists a parametrized datatype $A$ with
  $\sat \emptyset A t(\Gamma_1)$
  and
  $\sat \emptyset A t(\Gamma_2)$,
  but
  $\unsat \emptyset A t(\min(\Gamma_1, \Gamma_2))$.
\end{fact}
\begin{proof}
  Over two parameters $\alpha, \beta$, take
  $A := \guard{\tyeq \alpha \beta}\alpha$.
  Both $\Gamma_1 := \alpha:\ind, \beta:\sep$ and
  $\Gamma_2 := \alpha:\sep, \beta:\ind$ are admissible signatures for
  $A$.
  In particular, our inference rules can verify them: any $\Gamma'$ with
  $\Gamma' \vdash \alpha = \beta$ has
  $(\alpha:\sep, \beta:\sep)$.
  However, the minimum signature $\alpha:\ind, \beta:\ind$ is not valid for $A$.
\end{proof}

Note the example in the above proof can be represented as an OCaml
datatype as follows.
\begin{minted}{ocaml}
type (_, _) strange_eq =
  | Strange_refl : 'a -> ('a, 'a) strange_eq [@@unboxed]
\end{minted}

This results shows that some limitations of the system are not due to
our typing rules, but a fundamental lack of expressiveness of the
current modes and mode signatures as a specification/reasoning
language. We would need a richer language of modes, keeping track of
equalities between types, to hope to get a principal system.
%We discuss this in the Future Work section.

\begin{lemma}
  If $\Sigma; \Gamma_1 \vdash \tau : m$
  and $\Sigma; \Gamma_2 \vdash \tau : m$
  and $\CF{\tau}$
  then $\Sigma; \min(\Gamma_1, \Gamma_2) \vdash \tau : m$
\end{lemma}

\begin{proof}
  This proof is done by induction on the two derivations at once, but
  we need to use the Cut Elimination Lemma~\ref{lem:cut-elimination}
  first to be able to assume, by inversion/syntax-directedness,
  that the rules on both sides are the same.
\end{proof}

\begin{corollary}[Constraint-free principality]
  ~\\
  If $\Sigma; \Gamma \vdash \tau : m$ with $\CF{\tau}$,
  then there exists a minimal context $\Gamma_{\mathsf{min}}$
  such that $\Sigma; \Gamma_{\mathsf{min}} \vdash \tau : m$.
  ~\\
  If $\Sigma_0 \vdash \sigma \dashv \Sigma$ with $\CF{\sigma}$,
  then there exists a minimal signature $\Sigma_{\mathsf{min}}$
  such that $\Sigma_0 \vdash \sigma \dashv \Sigma_{\mathsf{min}}$.
\end{corollary}

\end{document}